\documentclass[11pt,a4paper]{article}
\usepackage{silence} 
\usepackage{jheppub}
\usepackage{graphicx}
\usepackage{soul}
\usepackage{cancel}
\usepackage[table]{xcolor}
\usepackage{orcidlink} 
\usepackage{enumitem}
\usepackage{amssymb,amsmath,mathrsfs,bm,bbm,mathtools} 
\usepackage{amsthm}
\definecolor{ao(english)}{rgb}{0.0, 0.5, 0.0}
\hypersetup{colorlinks=true,linkcolor=blue,urlcolor=blue,citecolor=ao(english)}
\usepackage{pifont} 
\usepackage{changepage}
\usepackage{stackrel}
\usepackage{multirow}
\usepackage[normalem]{ulem} 

\newcommand\nn{\nonumber\\}

\newcommand{\bma}{\left(\begin{array}}
	\newcommand{\ema}{\end{array}\right)}
\newcommand{\be}{\begin{equation}}
	\newcommand{\ee}{\end{equation}}
\newcommand{\ben}{\begin{equation*}}
	\newcommand{\een}{\end{equation*}}
\newcommand{\ba}{\begin{eqnarray}}
	\newcommand{\ea}{\end{eqnarray}}
\newcommand{\ban}{\begin{eqnarray*}}
	\newcommand{\ean}{\end{eqnarray*}}
\newcommand{\bs}{\begin{subequations}}
	\newcommand{\es}{\end{subequations}}
\newcommand{\bc}{\begin{center}}
	\newcommand{\ec}{\end{center}}
\newcommand{\ve}{\varepsilon}

\newcommand{\au}[2]{#1.~#2}
\newcommand{\arX}[1]{\href{http://arxiv.org/abs/#1}{{\cob arXiv:#1}}}
\newcommand{\oarX}[1]{\href{http://arxiv.org/abs/#1}{{\cob #1}}}

\newcommand{\book}[5]{\emph{#1}, #2, #3, #4 (#5)}
\newcommand{\books}[4]{\emph{#1}, #2, #3 (#4)} 
\newcommand{\doin}[6]{\href{http://dx.doi.org/#1}{{\cob {\it #2 #3} {\bf #4} (#6) #5}}}
\newcommand{\doinn}[5]{\href{http://dx.doi.org/#1}{{\cob {\it #2} {\bf #3} (#5) #4}}}
\newcommand{\doij}[5]{\href{http://dx.doi.org/#1}{{\cob {\it #2} {\bf #3} (#5) #4}}}

\newcommand{\procm}[6]{in \emph{#1}, #2 (eds.), #3, #4, #5 (#6)}
\newcommand{\tia}[1]{\textit{#1},}

\newcommand{\boxd}[1]{\boxed{\phantom{\Biggl(}#1\phantom{\Biggl)}}}

\renewcommand{\leq}{\leqslant}
\renewcommand{\geq}{\geqslant}
\newcommand{\Eq}[1]{(\ref{#1})}
\newcommand{\Eqq}[1]{eq.~(\ref{#1})}

\def\rme{e}
\def\rmd{d}
\def\rmi{i}

\def\Re{\text{Re}}
\def\Im{\text{Im}}

\def\a{\alpha}

\def\de{\delta}

\def\g{\gamma}
\def\la{\lambda}

\def\e{\epsilon}
\def\ve{\varepsilon}

\def\om{\omega}
\def\G{\Gamma}

\def\vp{\varphi}

\def\B{\Box}

\def\cA{\mathcal{A}}

\def\cD{\mathcal{D}}

\def\cG{\mathcal{G}}

\def\cL{\mathcal{L}}
\def\cM{\mathcal{M}}

\def\cO{\mathcal{O}}

\def\cob{\color{blue}}

\newtheorem{theo}{Theorem}

\def\pl{p}
\def\kl{k}
\def\pe{p_\textsc{e}}
\def\ke{k_\textsc{e}}
\def\pL{p_\textsc{l}}

\begin{document}
	
\renewcommand{\thefootnote}{\fnsymbol{footnote}}
	
\title{Amplitude prescriptions in field theories with complex poles} 
	
\author[a,b]{Damiano Anselmi\,\orcidlink{0000-0001-6674-1328},}
\emailAdd{damiano.anselmi@unipi.it}
\affiliation[a]{Dipartimento di Fisica ``E.~Fermi'', Largo B.~Pontecorvo 3, 56127 Pisa, Italy}
\affiliation[b]{INFN, Sezione di Pisa, Largo B.~Pontecorvo 3, 56127 Pisa, Italy}

\author[c,d,e]{Fabio Briscese\,\orcidlink{0000-0002-9519-5896},}
\emailAdd{fabio.briscese@uniroma3.it}
\affiliation[c]{Dipartimento di Architettura, Università Roma Tre, Via Aldo Manuzio 68L, 00153 Rome, Italy}
\affiliation[d]{Istituto Nazionale di Alta Matematica Francesco Severi, Gruppo
Nazionale di Fisica Matematica, Piazzale Aldo Moro 5, 00185 Rome, Italy}
\affiliation[e]{Istituto Nazionale di Fisica Nucleare, Sezione di Roma 3, Via della Vasca Navale 84, 00146 Rome, Italy}

\author[f,*]{Gianluca Calcagni\,\orcidlink{0000-0003-2631-4588}\note{Corresponding author.}}
\emailAdd{g.calcagni@csic.es}
\affiliation[f]{Instituto de Estructura de la Materia, CSIC, Serrano 121, 28006 Madrid, Spain}

\author[g,h]{and Leonardo Modesto\,\orcidlink{0000-0003-2783-8797}}
\emailAdd{leonardo.modesto@unica.it}
\affiliation[g]{Dipartimento di Fisica, Universit\`a di Cagliari, Cittadella Universitaria, 09042 Monserrato, Italy}
\affiliation[h]{INFN, Sezione di Cagliari, Cittadella Universitaria, 09042 Monserrato, Italy}

\abstract{In the context of field theories with complex poles, we scrutinize four inequivalent ways of defining the scattering amplitudes, each forfeiting one or more tenets of standard quantum field theory while preserving the others: (i) a textbook Wick rotation by analytic continuation of the external momenta from Euclidean to Lorentzian signature (no optical theorem), (ii) the Lee--Wick--Nakanishi prescription, integrating along a certain contour in the complex energy plane (no Lorentz invariance), (iiii) the fakeon prescription, where, in addition, spatial momenta are integrated on a complex path defined by the locus of singularities of the loop integrand (no analyticity of the amplitude) and (iv) to work directly on Minkowski spacetime, which violates the optical theorem and also bars power-counting renormalizability. In general, mixed Euclidean-Lorentzian prescriptions for internal and external momenta in loop integrals break Lorentz invariance, regardless of the type of masses involved. We conclude that, of the above four options, only the fakeon prescription is physically viable and can have applications to quantum gravity.}


\maketitle
\renewcommand{\thefootnote}{\arabic{footnote}}


\section{Introduction}\label{sec1}

Amid the exploration of new quantum field theories and their properties,
complex poles have made their appearance both as a general proof of concept \cite%
{Veltman:1963th,Yamamoto:1969vb,Yamamoto:1970gw,Nakanishi:1972pt,Bollini:1998hj,Mannheim:2018ljq,Buoninfante:2025klm} and in as diverse physical scenarios as modified electrodynamics
(Lee--Wick theory) \cite%
{Lee:1969fy,Lee:1970iw,Lee:1969zze,Cutkosky:1969fq,Nakanishi:1971jj}, a
modified Higgs sector \cite{Jansen:1993jj,Koshelev:2020fok}, standard gauge
theories and quantum chromodynamics (QCD) \cite%
{Nakanishi:1975aq,Dudal:2005na,Dudal:2007cw,Dudal:2008sp,Baulieu:2009ha,Dudal:2011gd,Capri:2012hh}
and quantum gravity \cite%
{Holdom:2015kbf,Modesto:2015ozb,Modesto:2016ofr,Mannheim:2020ryw,Liu:2022gun,deBrito:2023pli,Tokareva:2024sct,Asorey:2024mkb,BrCa}
(see also \cite{Frasca:2022gdz,Addazi:2024qcv}). Interest in complex poles
has been dictated by the fact that, according to the evidence available in
all known accounts, they respect perturbative unitarity when they are
present in complex-conjugate pairs. This can be exciting news because, on
one hand, one can construct higher-derivative theories based on complex
pairs, which have the improved renormalizability typical of
higher-derivative models without the hindrance of ghosts. On the other hand,
in nonlocal theories such as nonlocal quantum gravity \cite%
{Modesto:2017sdr,Buoninfante:2022ild,BasiBeneito:2022wux,Koshelev:2023elc}
and fractional quantum gravity \cite{BrCa,Calcagni:2021aap,Calcagni:2022shb}%
, complex-conjugate pairs may be hidden in operators with infinitely many
derivatives \cite{Pais:1950za} but, again, they do not spoil
renormalizability and unitarity.

It may come as an unpleasant surprise that all these models suffer from a
problem which is not as well known as it would deserve: if scattering amplitudes are not defined extra carefully, Lorentz invariance is broken at the quantum level. Nakanishi \cite{Nakanishi:1971jj} was the
first to point it out in the context of Lee--Wick theory. Let $\mathcal{P}$
and $\mathcal{P}^*$ denote the particles in a pair with complex-conjugate masses. For our purpose, we can assume that such masses are purely complex, $\pm\rmi M^2$ with $M$ real.
According to \cite{Nakanishi:1971jj}, the breaking of Lorentz invariance already
occurs at the level of a bubble diagram with $\mathcal{P}$ in an
internal line with momentum $k$ and propagator\footnote{We use
signature $(-,+,\dots,+)$.}
\be\label{propc1}
\tilde G(-k^2,\rmi M^2) = -\frac{\rmi}{k^2+\rmi M^2}\,,\qquad M\in\mathbb{R}\,,
\ee
and $\mathcal{P}^*$ in the other internal line with momentum $p+k$ and propagator $-\rmi/[(p+k)^2-\rmi M^2]$. 
 This result does not change if one takes all combinations of bubble diagrams with the particles $\mathcal{P}$ and $\mathcal{P}^*$, namely, the sum of four bubbles with internal lines corresponding to $\mathcal{P}$-$\mathcal{P}^*$, $%
\mathcal{P}^*$-$\mathcal{P}$, $\mathcal{P}$-$\mathcal{P}$ and $\mathcal{P}^*$%
-$\mathcal{P}^*$, even if the total is equivalent to a single bubble
diagram with a Lorentz-invariant integrand. In this case, internal lines correspond to the propagators of the pairs $-\rmi/(k^4+M^4)$ and $-\rmi/[(p+k)^4+M^4]$ instead of the single particles.

In sections~\ref{sec2} and \ref{sec3}, we recall in a novel, simplified way that the problem arises because, if one defines the physical amplitude of a theory with complex masses as the analytic continuation of the Euclidean amplitude, in the same way as in standard quantum field theory (QFT), then the optical theorem is violated. This problem is what led to the Lee--Wick prescription in the first place \cite{Lee:1969fy,Lee:1970iw,Lee:1969zze}. Then, we show that Lorentz symmetry is manifestly violated when one integrates the loop energy along the Lee--Wick complex path while keeping the space components of the loop momentum real. We refer to this procedure as the Lee--Wick--Nakanishi (LWN) prescription, to distinguish it from Lorentz preserving approaches to Lee--Wick theories, such as the one of \cite{Cutkosky:1969fq}. This symmetry breaking  depends neither on the details of the action nor on the choice of contour. Whether a complex pole is isolated or paired to its conjugate is also irrelevant.


This sounds like a fatal blow to the above constructions in quantum gravity and in other contexts, including relatively orthodox models of gauge theories and quark confinement
in QCD \cite{Dudal:2005na,Dudal:2007cw,Dudal:2008sp,Baulieu:2009ha,Dudal:2011gd}. Fortunately, the cure has already been known since a while and the key is how one defines loop integrals. As we show in an example with \emph{real} masses, Lorentz violation is actually not a physical issue intrinsic to complex masses but a general technical problem arising when integrands have denominators with finite real and complex parts. These integrals are of the type 
\be\label{eq12}
\int\rmd x\,\frac{f(x)}{x+\rmi y}\,,\qquad \mathbb{R}\ni y\neq 0\,,
\ee
where $y$ is not infinitesimal. In theories with complex-valued masses, the LWN prescription essentially leads to the same problems as for loop integrals with Euclidean internal momenta and Lorentzian external ones. If one moves forward with respect to the plain LWN prescription, Lorentz invariance may be restored at all perturbative orders without altering unitarity. One manifestly Lorentz invariant approach is the one by Cutkosky--Landshoff--Olive--Polkinghorne (CLOP) \cite{Cutkosky:1969fq}, which leaves room for ambiguities at high orders. A different approach is to restore Lorentz invariance by deforming the integrals on spatial momenta to complex paths \cite%
{Anselmi:2017yux}. The latter is one of the ingredients of the fakeon or Anselmi--Piva (AP)
prescription \cite%
{Anselmi:2017yux,Anselmi:2017lia,Anselmi:2018bra,Anselmi:2019rxg,Anselmi:2021hab,Anselmi:2022toe}, a general recipe to remove ghost modes from the physical spectrum of quantum field theories. 
 This seemingly technical point is sufficient to solve the physical problem of avoiding violations of one of the best constrained symmetries in Nature. However, its involved formulation has so far contributed to the aura of mystery around the topic of Lorentz symmetry in the presence of particles with complex masses.

Section \ref{sec4} of the present paper aims at filling this gap and presenting the AP prescription in as pedagogical terms as possible. We do this with a one-loop example with complex poles in $D=1+1$ spacetime dimensions, showing how a specific path in the complex plane of spatial momentum restores Lorentz invariance of the bubble diagram. In the same section, we also describe a notable simplification for calculations with this procedure, first used in \cite{Anselmi:2017yux,Liu:2022gun}. 

The net conclusion we can draw from what is found in this paper are that, when we define a scattering amplitude in QFT in the presence of complex poles, we face a choice among the following four alternatives (table \ref{tab1}):
\begin{itemize}
\item \emph{By-the-book prescription}: the calculation is done with Euclidean internal and external momenta. At the end, the amplitude is Wick rotated by means of the usual analytic continuation. Lorentz symmetry is preserved, analyticity holds but the optical theorem is violated. This is the standard prescription in ordinary QFT with real poles.
\item \emph{LWN prescription}: the calculation is done with Lorentzian internal and external momenta. The loop energy $k^0$ is integrated over the Lee--Wick path, which we dub $\G_{\rm LW}$ from now on, and the space components of the loop momenta are integrated over real values. The optical theorem holds but Lorentz symmetry and analyticity are violated.
\item \emph{Fakeon or AP prescription}: the same as for the LWN prescription but spatial internal momenta are integrated on a specific complex path; equivalently, the calculation is first done with Euclidean internal and external momenta and then one averages the analytic continuations around branch cuts (this procedure is called \textit{average continuation}). The optical theorem and Lorentz symmetry hold but analyticity is violated.
\item \emph{Direct Minkowski prescription}: the calculation is done directly with Lorentzian loop momenta (all integrated on $\mathbb{R}^D$) and Lorentzian external momenta. Then, Lorentz invariance and analyticity are preserved but the optical theorem and the locality of counterterms are violated (see \cite{Aglietti:2016pwz} and appendix~\ref{appA}).
\end{itemize}
In physical terms, the least harmful violation is that of analyticity, which is a mere mathematical property. Lorentz symmetry and unitarity are, in contrast, requirements with wide physical consequences. Having local counterterms is also important, since it is related to the possibility to study renormalizability with power counting and to renormalize the theory explicitly with the Bogoliubov--Parasiuk--Hepp--Zimmermann (BPHZ) scheme \cite{Lowenstein:1975ug,Piguet:1995er}.
\begin{table}[ht]
\begin{center}
\footnotesize
\begin{tabular}{|l|c|cccc|}\hline\hline
{ Prescription} & { Features } & { Lorentz}    & { Optical} & { Power}    & { Analyticity} \\
                &              & { invariance} & { theorem} & { counting} &                \\\hline
By-the-book     & $\pe$, $\ke$; $\bm{k}\in\mathbb{R}^{D-1}$  & \ding{51} &      & \ding{51} & \ding{51} \\
& E$\to$L an.\ cont.\  & & & & \\ & & & & & \\
Lee--Wick--Nakanishi   & $\pl$, $\kl$; $k^0\in{\Gamma}_{\textrm{LW}}$; $\bm{k}\in\mathbb{R}^{D-1}$ &  & \ding{51} & \ding{51} & \\
& L, patch-wise  & & & & \\ & & & & & \\
Anselmi--Piva          & $\pl$, $\kl$; $k^0\in{\Gamma}_{\textrm{LW}}$; $\bm{k}\in\mathbb{C}^{D-1}$ & \ding{51} & \ding{51} & \ding{51}  &           \\
& $\equiv$  & & & & \\
& $\pe$, $\ke$; $\bm{k}\in\mathbb{R}^{D-1}$  & & & & \\
& E$\to$L avg.\ cont.  & & & & \\ & & & & & \\
Direct Minkowski    & $\pl$, $\kl$; $k^0\in\mathbb{R}$; $\bm{k}\in\mathbb{R}^{D-1}$   & \ding{51} &  &           & \ding{51} \\
& L  & & & & \\\hline
\end{tabular}
\caption{\label{tab1} Prescriptions on scattering amplitudes in QFTs with complex poles. ``E'' stands for Euclidean and ``L'' for Lorentzian.}
\end{center}
\end{table}

The problem of the LWN prescription in preserving Lorentz invariance is that it forces $D-1$ of the integration variables $k^\mu$ ($\bm{k}$) to stay real, while only one of them ($k^0$) is allowed to take complex values. Lorentz transformations mix these variables and thus make it impossible to satisfy the LWN condition on $\bm{k}$ in all references frames. 
In section~\ref{sec4}, we note that it is possible to restore Lorentz invariance by deforming the integration domain of spatial momenta from real to complex regions (AP prescription). 

We denote the internal and external momenta in Lorentzian signature with $\kl$ and $\pl$, respectively,
\be
\kl^\mu=(k^0,\bm{k})\,,\quad \kl^2=-(k^0)^2+|\bm{k}|^2\,,\qquad \pl^\mu=(p^0,\bm{p})\,,\quad \pl^2=-(p^0)^2+|\bm{p}|^2\,, 
\ee
reserving the symbols $\ke$ and $\pe$ for the internal and external momenta in Euclidean signature:
\be
\ke^\mu=(k_D,\bm{k})\,,\quad \ke^2=k_D^2+|\bm{k}|^2\,,\qquad \pe^\mu=(p_D,\bm{p})\,,\quad \pe^2=p_D^2+|\bm{p}|^2\,.
\label{Euclideanmomenta}
\ee
Euclidean amplitudes will be denoted by the symbol $\cM$, while Lorentzian amplitudes will carry a subscript indicating the prescription adopted.

Appendix~\ref{appA} recalls the problems of the direct Minkowski prescription, while appendix~\ref{appB} revisits how modes with complex masses enter the optical theorem. Appendix~\ref{appC} reports the calculation as well as the analytic and average continuations of the one-loop bubble amplitude of a prototypical $\phi^3$ model. An example of Lorentz violation in a model with real masses and mixed Euclidean-Lorentzian signature for internal and external momenta is given in appendix~\ref{appD}.


\section{Optical theorem and analyticity}\label{sec2} 
	
In this section, we argue that in the presence of a branch cut the relation between Euclidean and Minkowskian amplitudes in QFT is not unique \textit{a priori}: two options are viable, with different impacts on the physical contents of the theory, as revealed by the optical theorem. The standard approach involves the analytic continuation on a specific side of the cut. A non-analytic alternative, henceforth referred to as the average continuation, is to average the two analytic continuations around the cut.
	
With the aim of showing what we mean, we consider the  ordinary bubble diagram in the massless $\la\phi^3$ theory. Using dimensional regularization, the amplitude in Euclidean space
\be
\cM(\pe^2)=\frac{\la^2}{2}\int\frac{\rmd^D\ke}{(2\pi)^D}\,\frac{1}{\ke^2}\frac{1}{(\ke+\pe)^2}\,,
\ee
is
\be
\cM(\pe^2)= \frac{\la^2}{(4 \pi)^2\ve}-\frac{\la^2}{32\pi^2}\left[\ln\left(\frac{\pe^2}{4\pi}\right)+\gamma_{\textsc{em}}-2\right]+O(\ve)\,,
\label{ME}
\ee
where $D=4-\ve$ and $\gamma_{\textsc{em}}$ is the Euler--Mascheroni constant, and the logarithm is taken in its principal branch, so that it is real for real positive values of its argument.

The dependence of the amplitude (\ref{ME}) on the external momentum is condensed in the variable $\pe^2$, with $\pe^2 \geq 0$.
We denote the analytic continuation of $\cM(\pe^2)$ from real positive values $\pe^2$ to $z\in\mathbb{C}$ by $\cM(z)$.

Our task is to switch (\ref{ME}) to the Lorentzian signature and define the physical amplitude in Minkowski spacetime. This demands to extend the function (\ref{ME}) from real positive to real negative $z$, which can be achieved in two basic ways:\footnote{Turning on a mass, as we do in section~\ref{sec3}, simply shifts the branch point from $z=0$ to some other value ($z=-2$ there), but does not change the conclusions below.}
\begin{enumerate}
\item[(i)] By analytic continuation $p_D \rightarrow -\rmi p^0$ with $p^0\in \mathbb{R}^+_0$ of the function $\cM(\pe^2)$ from below the branch cut:
\be\label{bbmain}
\boxd{\cM_{\textrm{bb}}(\pl^2)\coloneqq \cM(\pl^2-\rmi\e)\,,}
\ee
where ``bb'' stands for ``by the book,'' since this is the standard textbook way to define the Lorentzian amplitude. As (\ref{ME}) has a branch point at $z=\pe^2=0$, its analytic continuation is encoded in that of the function $\ln(z)$ from real positive to real negative $z$ passing below the branch point at $z=0$. The result is
\be
\cM_{\textrm{bb}}(p^2)= \frac{\la^2}{(4 \pi)^2\ve}-\frac{\la^2}{32\pi^2}\left[{\rm Ln}\left(\frac{p^2-\rmi\e}{4\pi}\right)+\gamma_{\textsc{em}}-2\right]+O(\ve)\,,
\label{Mph}
\ee
where $\textrm{Ln}(z)$ is the principal branch of the complex logarithm, so that $\textrm{Ln}(z)\coloneqq \ln|z|+\rmi\,{\rm Arg}\,z$, with $-\pi<{\rm Arg}\,z\leq \pi$, and $\textrm{Ln}\,z^2= 2\,\textrm{Ln}\,z$ if $\Re\, z>0$, while $\textrm{Ln}(z^2)= 2\, \textrm{Ln}(-z)$ if $\Re\, z<0$.

This procedure defines a complex amplitude with a branch-cut at $p^2=0$, corresponding to the threshold of production of real intermediate massless particles, as we discuss below. Equivalently, one can analytically continue by passing above the branch point at $z=0$, obtaining the complex conjugate amplitude 
\be
\cM^*_{\textrm{bb}}(p^2)= \cM(\pl^2+\rmi\e)\,,
\label{Mph1}
\ee
which is useful for studying the optical theorem.
\item[(ii)] Alternatively, one can extend (\ref{ME}) to real negative $z$ by defining a complex, non-analytic function piecewise:
\be
\cM_{\textrm{AP}}(z)= 
\left\{ \begin{array}{ll} \cM(z) \qquad\quad\!\! \text{for } \quad \Re\, z> 0\\
	\cM(-z) \qquad \text{for } \quad \Re\, z<0 		\end{array}\right. \,. \label{MF0}
\ee
Explicitly, (\ref{MF0}) can be rewritten as 
\be
\cM_{\textrm{AP}}(p^2)= \frac{\la^2}{(4 \pi)^2\ve}-\frac{\la^2}{32\pi^2}\left\{\frac{1}{2}\,  \textrm{Ln}\left[\frac{(p^2)^2}{4\pi}\right]+\gamma_{\textsc{em}}-2\right\}+O(\ve)\,.
\label{Mf}
\ee
		
When evaluated for real $p^2$, (\ref{Mf}) is real. For negative $p^2$,  it  coincides with the average of the two analytic continuations (\ref{Mph})-(\ref{Mph1}) around the branch point at $p^2=0$. As a result, the complex amplitude (\ref{Mf})  does not have any threshold of production of real intermediate states. This fact is related to the existence of purely virtual particles (``fakeons''), which  cannot be produced in scattering processes as intermediate real states, as discussed below. 
\end{enumerate}
	
The two choices (i) and (ii) discussed above have different impacts on the optical theorem, which encodes the unitarity condition $SS^\dagger=1$ on the $S$ matrix  as a condition on the imaginary part of the complex amplitude. In the instances of interest to us, which is the bubble diagram, the optical theorem can be graphically written as
\be
2\,\text{Im}\left[ (-\rmi)\raisebox{-1mm}{\scalebox{2}{$-%
\hspace{-0.055in}\bigcirc\hspace{-0.055in}-$}}\right] =\int \mathrm{d}\Pi _{\rm f}\hspace{0.01in}%
\left\vert \raisebox{-1mm}{\scalebox{2}{$-\hspace{-0.035in}\langle$}}%
\right\vert ^{2},  \label{opt}
\ee
where the integral in the right-hand side is over the phase space of final states. A physical implication of this diagrammatic equation (which enforces unitarity at the one-loop level) is that, if the complex amplitude has a non-zero imaginary part, then the particles corresponding to the propagators of the internal loop momenta can be produced as real intermediate states. 

In the case (i), the amplitude (\ref{Mph}) has a branch point at $p^2=0$, and it has a non-zero imaginary part for real external momenta such that $p^2<0$. That means that particles corresponding to the internal lines in the bubble can be produced as intermediate states and must be in the spectrum of the theory. Therefore, if one aims at quantizing fake particles, e.g., modes which are excluded from the spectrum of the theory, one cannot define the complex amplitude by means of the procedure of analytic continuation of  (\ref{MF0}) defined in (i). However, the procedure described in (ii), corresponding to the ``average" between the two possible analytic continuations of (\ref{MF0}) around the branch point at $p^2=0$, gives the right answer for fake particles, as the amplitude (\ref{Mf}) is real for any real $p^2$, and the particles in the loop can never be produced as real intermediate states.

Switching to theories with complex poles, the conclusion is the same. In appendix~\ref{appB}, we show that, since the tree propagators 
\be
-\frac{\rmi M^2}{(p^2+m^2)^2+M^4}
\label{HDprop}
\ee
are purely imaginary, the contribution of the bubble diagram with circulating fake particles to the amplitude must necessarily be real, for consistency with the optical theorem (\ref{opt}): 
\be
2\,\textrm{Im}\,\cM_{\textrm{AP}}(p^2)= 0\,.
\ee

In conclusion, the simple example discussed in this section shows that consistency with the optical theorem demands that the complex poles must not be quantized via the standard analytic continuation (i) but as fake particles, e.g., by means of the averaging procedure (ii).


\section{Prescriptions preserving the optical theorem with complex masses}\label{sec3}

In this section, we compute the bubble diagram with the propagator (\ref{HDprop}). We show that if we apply the ``by the book'' procedure (i.e., calculate the diagram in the Euclidean framework and then define the Minkowskian amplitude by means of the analytic continuation),
the result is not real, in conflict with the optical theorem. If, instead, we use the AP recipe (average continuation), the result is real (hence agrees with the optical theorem), but not analytic. Moreover, we show that a new prescription emerges almost by accident through a subtle yet natural misstep in the calculation. 

It may help to view what follows as the generalization of the case discussed in the previous section to complex-conjugate masses. Consider the Lagrangian
\be\label{LWmodel}
\cL=-\frac12\,\phi\left[\left(\frac{\B-m^2}{M^2}\right)^2+1\right]M^2\phi-\frac{\la}{3!}\,\phi^3
\ee
in $D=4$ dimensions, where $m$ and $M$ are real masses. The tree-level propagator is
\ba
\hspace{-.5cm}\tilde \cG(-\kl^2)&=&-\frac{\rmi M^2}{(\kl^2+m^2)^2+M^4}
\ea 

The Euclidean bubble amplitude, evaluated in appendix \ref{appC1}, is 
\be\label{final}
\cM(\pe^2)=\frac{\la^2}{128\pi^2}\int_0^1\rmd x\,\ln\frac{[x(1-x)\pe^2+m^2]^2+(1-2x)^2M^4}{[x(1-x)\pe^2+m^2]^2+M^4}\,,
\ee
``ln'' being intended as the principal branch. Moreover, hereafter $\cM(z)$ denotes the analytic continuation of $\cM(\pe^2)$ from real positive values $\pe^2$ to $z\in\mathbb{C}$.

We readily note the following fun fact: the argument of the logarithm, which is nonnegative for all Euclidean momenta $\pe$, remains nonnegative when we replace $\pe$ with its Lorentzian version $\pl$, for every $\pl$. Thus, we may be tempted to believe that the analytic continuation of (\ref{final}) to Minkowski spacetime is
the one obtained by means of the replacement $\pe\rightarrow \pl$, which reads
\be\label{finalL}
\cM_{\textrm{fM}}(\pl^2)=\frac{\la^2}{128\pi^2}\int_0^1\rmd x\,{\rm Ln}\frac{[x(1-x)\pl^2+m^2]^2+(1-2x)^2M^4}{[x(1-x)\pl^2+m^2]^2+M^4}\,,
\ee
where ``fM'' stands for formal Minkowski. In appendix \ref{appC2}, we show that this expression is neither the analytic continuation of (\ref{final}), nor the average continuation of it. The simple but non-trivial replacement $\pe^2\to \pl^2$ in the integrand of \Eq{final}
is not the AP prescription studied in the literature, which is given by the average continuation
\be\label{APmain}
\boxd{\cM_{\textrm{AP}}(\pl^2)\coloneqq \frac12\left[\cM(\pl^2-\rmi\e)+\cM(\pl^2+\rmi\e)\right]\,.}
\ee

We have $\cM_{\textrm{AP}}(\pl^2)=\cM(\pl^2)$ in the zero-mass case \Eq{MF0} and \Eq{Mf} of the previous section. Indeed, a Euclidean amplitude is real for any positive $z=\pe^2$: $\cM(z>0)\in\mathbb{R}$. If we further assume the reflection symmetry $\cM(-z)=\cM(z)$, then there is no branch cut on the whole real line. Hence the $\pm i\epsilon$ displacements become irrelevant and
\be
\cM_{\rm AP}(\pl^2) \stackrel{\textrm{\tiny \Eq{APmain}}}=\cM(\pl^2)\,. 
\ee
If $\cM(z)\neq\cM(-z)$ for real $z$, then one cannot draw this conclusion. The specific model discussed in this section has $\cM_{\rm AP}(\pl^2) \neq \cM(\pl^2)$. 

In general, the analytic continuation may break the symmetry properties
of the integrand in the amplitude.
In appendix~\ref{appC2}, we show this by analytically continuing from the domain $0<z<2$, which is asymmetric with respect to the imaginary axis. At the time of writing, we do not have an argument showing that the condition $\cM(z)=\cM(-z)$ for $z\in\mathbb{R}$ is also necessary.


\section{Lorentz breaking and Lorentz restoring}\label{sec4}

In this section, we demonstrate that the LWN prescription violates Lorentz invariance and show that the symmetry can be restored by deforming the integration domain in a suitable way. 

According to the LWN prescription, the scattering amplitude should be evaluated by integrating the loop energy along the Lee-Wick complex path $\G_{\textrm{LW}}$ (see figure \ref{fig1}), while keeping the spatial momenta real. If the result were Lorentz invariant, we would have, comparing two reference frames:
\be
\int_{\G_{\textrm{LW}}\times\mathbb{R}^{D-1}}\rmd^Dk\,f(k,p) \stackrel{?}{=} \int_{\G_{\textrm{LW}}^\prime\times\mathbb{R}^{D-1}}\rmd^Dk'\,f(k',p')\,.
\ee
However, a finite Lorentz transformation in momentum space
\be\label{eq2}
k^{\mu\prime} = \Lambda_\nu^\mu k^\nu = \Lambda_0^\mu k^0+\Lambda_j^\mu k^j
\ee
does not turn the integration domain into an equivalent one. In fact, since $k^0$ is complex and both the $k^i$ and the Lorentz matrix elements are real, then 
\be\label{eq2b}
k^{i\prime} = \Lambda_0^i k^0+\Lambda_j^i k^j
\ee
must be complex. This violates the LWN prescription that would require $\bm{k}'\in\mathbb{R}^{D-1}$ independently of the frame choice.

To be specific, observe that integrating the loop energy along the LW path is equivalent to integrating it along the real line, which gives the ``direct-Minkowski'' amplitude (\ref{MtotappA}), and adding the contributions of certain residues (see figure \ref{fig1}). The direct-Minkowski amplitude is manifestly Lorentz invariant (appendix \ref{appA}). On the other hand, the residue associated with some pole $k^{0}=$ $\ \bar{k}^{0}$ gives,
in a specifice Lorentz frame, a contribution of the form%
\begin{equation}
\int_{\Gamma }dk^{0}\int_{\mathbb{R}^{3}}d^{3}\bm{k~}f(k^{0},\bm{k}%
,p)\bm{=}\int_{[0,1]\times \mathbb{R}^{3}}d\tau d^{3}\bm{k~}\dot{k}%
^{0}(\tau )f[k^{0}(\tau ),\bm{k},p]\,,  \label{l1}
\end{equation}%
where $\Gamma$ is a small circle around $\bar{k}^{0}$ parametrized by a real parameter $\tau$ and the dot denotes the derivative with respect to $\tau$. In another frame, we have the same expression with primes everywhere. We are assuming that the integrand is a
scalar, i.e., $f(k^{0\prime },\bm{k}^{\prime },p^{\prime
})=f(k^{0},\bm{k},p)$. If there existed a change of variables $%
\tau ,\bm{k}\rightarrow \tau ^{\prime },\bm{k}^{\prime }$ that
induces the Lorentz transformation, then we could write%
\begin{eqnarray}
\hspace{-.5cm}\int_{\Gamma ^{\prime }}dk^{0\prime }\int_{\mathbb{R}^{3}}d^{3}\bm{k}%
^{\prime }\bm{~}f(k^{0\prime },\bm{k}^{\prime },p^{\prime }) &%
\bm{=}&\int_{[0,1]\times \mathbb{R}^{3}}d\tau ^{\prime }d^{3}\bm{k}%
^{\prime }\bm{~}\dot{k}^{0\prime }(\tau ^{\prime })f(k^{0\prime }(\tau
^{\prime }),\bm{k}^{\prime },p^{\prime })  \nonumber \\
\hspace{-.5cm}&=&\int_{[0,1]\times \mathbb{R}^{3}}d\tau \,d^{3}\bm{k~}\dot{k}^{0\prime
}(\tau ^{\prime })\left\vert \frac{\partial (\tau ^{\prime },\bm{k}%
^{\prime })}{\partial (\tau ,\bm{k})}\right\vert f[k^{0}(\tau),\bm{k},p]\,.  \label{l2}
\end{eqnarray}%
Since the function $f$ is arbitrary, to have (\ref{l1}) equal to (\ref{l2})
we need
\begin{equation}
\left\vert \frac{\partial (\tau ^{\prime },\bm{k}^{\prime })}{\partial
(\tau ,\bm{k})}\right\vert =\frac{\dot{k}^{0}(\tau )}{\dot{k}^{0\prime
}(\tau ^{\prime })}.  \label{forma}
\end{equation}%
The left-hand side is real, while the right-hand side is complex, since the
curves $\Gamma$ and $\Gamma^{\prime }$ are located away from the real
axis. This shows that the LWN\ prescription is not manifestly Lorentz
invariant. The explicit calculation below shows that it does violate Lorentz
symmetry.

The argument does not apply to the case where the poles are located on the
real axis, since the right-hand side of (\ref%
{forma}) is real when the curves $\Gamma$ and $\Gamma^{\prime }$ are
arbitrarily close to the real axis.

Now we show that Lorentz invariance is broken by the LWN prescription but can be restored by deforming the domain of spatial momenta to complex values. We also show that the domain deformation is equivalent to the average continuation. Besides presenting the domain deformation in a simpler way, we give a new proof of this equivalence. 

The bubble diagram with propagators (\ref{HDprop}) contributes to the amplitude via the integral 
\begin{equation}
	{\cal M}(\pl)=-\frac{\rmi\la^2}{2}\int_{\G_\textrm{LW}} \frac{\rmd^{0}\kl}{2\pi}\int_{\mathbb{R}^{D-1}} \frac{\rmd^{D-1}\bm{\kl}}{(2\pi)^{D-1}}\underbrace{\frac{M^2}{(\kl^{2})^{2}+M^{4}}}_{\rm (I)}\underbrace{\frac{M^2}{[(\kl+\pl)^{2}]^{2}+M^{4}}}_{\rm (II)} 
    \,, \label{BAP}
\end{equation}%
where the momenta $k$ and $p$ are Lorentzian, with the caveat that the integral on the loop energy $k^0$ is performed along the Lee--Wick integration path depicted in figure \ref{fig1}.
\begin{figure}[t]
\begin{center}
\includegraphics[width=12truecm]{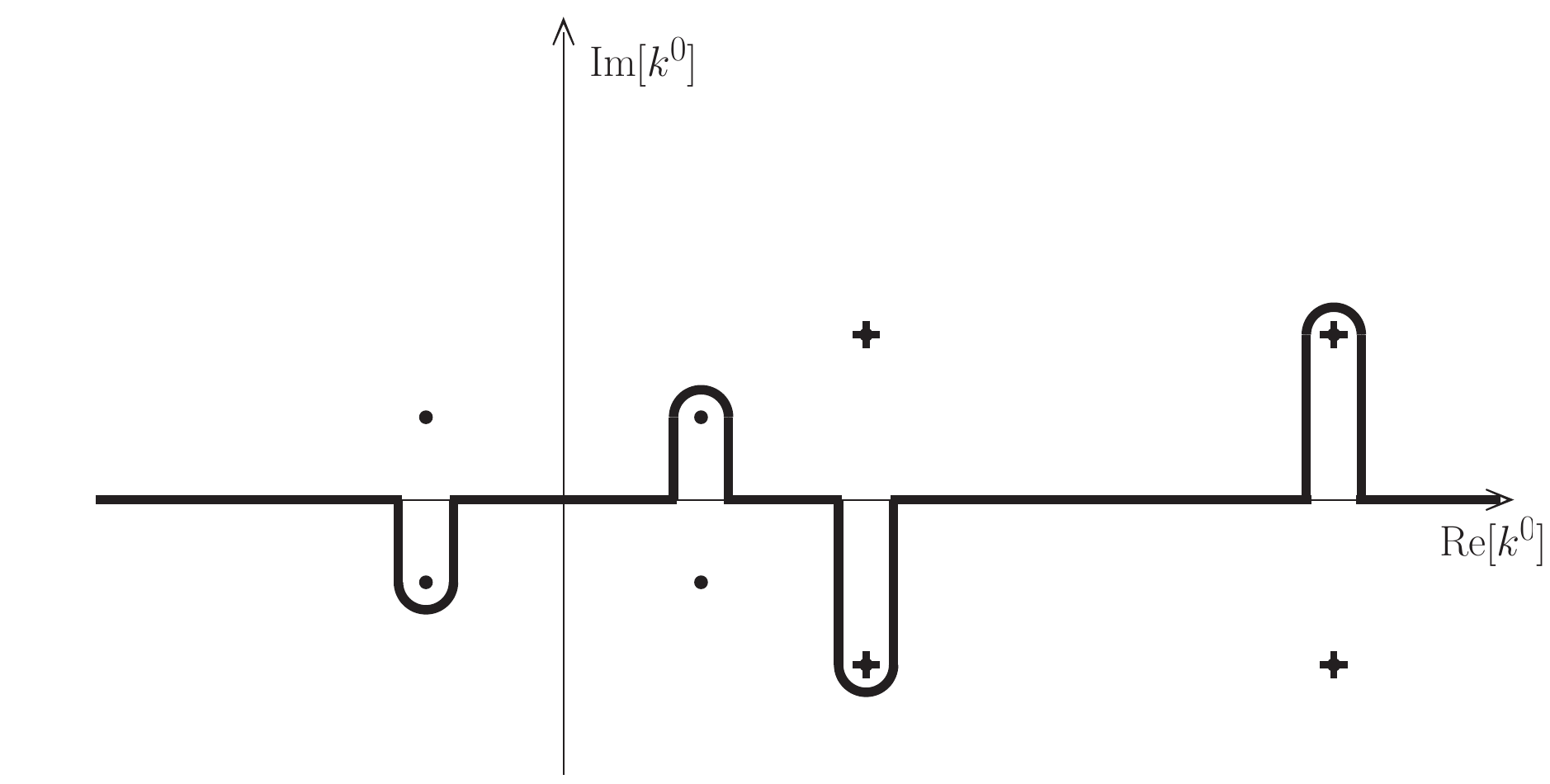}
\end{center}
\caption{The Lee--Wick integration path $\G_{\textrm{LW}}$. The dots are the $p^0$-independent poles, while the crosses denote the $p^0$-dependent ones.}
\label{fig1}
\end{figure}

The propagators have four poles each at $k^0=\bar k_{1}^{0\pm},\bar k_{1}^{0\pm*},\bar k_{2}^{0\pm},\bar k_{2}^{0\pm*}$:
\begin{equation*}
	\text{(I) }\left\{ 
	\begin{tabular}{l}
		$\bar k_{1}^{0\pm}=\pm \Omega _{\bm{k}},$ \\ 
		$\bar k_{1}^{0\pm*}=\pm \Omega _{\bm{k}}^{\ast },$%
	\end{tabular}%
	\right. \quad \quad \text{(II) }\left\{ 
	\begin{tabular}{l}
		$\bar k_{2}^{0\pm}=-p^{0}\pm\Omega _{\bm{k}+\bm{p}},$ \\ 
		$\bar k_{2}^{0\pm*}=-p^{0}\pm\Omega _{\bm{k}+\bm{p}}^{\ast },$%
	\end{tabular}%
	\right. ,
\end{equation*}%
respectively, where $\Omega _{\bm{q}}\coloneqq \sqrt{\bm{q}^{2}-\rmi M^{2}}$. 

We write each propagator as the sum of simple poles by partial fraction decomposition. Then we expand the integrand and integrate each term on $k^0$ with the Lee--Wick contour prescription, using the residue theorem. We do not need to calculate the $\bm{k}$ integral explicitly.

The result of the operations just mentioned is a linear combination of 
\begin{eqnarray*}
	&&\text{(A) }=\frac{1}{\Omega _{\bm{k}}+\Omega _{\bm{k}+\bm{p}%
		}^{\ast }\pm p^{0}},\quad \text{(B) }=\frac{1}{\Omega _{\bm{k}}^{\ast
		}+\Omega _{\bm{k}+\bm{p}}\pm p^{0}}, \\
	&&\text{(C) }=\frac{1}{\Omega_{\bm{k}}+\Omega _{\bm{k}+\bm{p}%
		}\pm p^{0}},\quad \text{(D) }=\frac{1}{\Omega _{\bm{k}}^{\ast }+\Omega _{%
			\bm{k}+\bm{p}}^{\ast }\pm p^{0}},
\end{eqnarray*}%
to be further integrated over ${\bm k}\in\mathbb{R}^{D-1}$. More precisely, we find
\begin{eqnarray}
-32\lambda^{-2}{\cal M}(p)&=&\int_{\cD^\prime_{\bm k}}\frac{1}{\Omega _{\bm{k}}\Omega _{\bm{k}+\bm{p}}(\Omega _{\bm{k}}+\Omega _{\bm{k}+\bm{p}}+p^{0})}-\int_{\cD_{\bm k}} \frac{1}{\Omega _{\bm{k}}\Omega _{\bm{k}+\bm{p}}^*(\Omega _{\bm{k}}+\Omega _{\bm{k}+\bm{p}}^*+p^{0})}\notag\\
&-&\int_{\cD_{\bm k}^*}\frac{1}{\Omega _{\bm{k}}^*\Omega _{\bm{k}+\bm{p}}(\Omega _{\bm{k}}^*+\Omega _{\bm{k}+\bm{p}}+p^{0})}+\int_{\cD^{\prime *}_{\bm k}}\frac{1}{\Omega _{\bm{k}}^*\Omega _{\bm{k}+\bm{p}}^*(\Omega _{\bm{k}}^*+\Omega _{\bm{k}+\bm{p}}^*+p^{0})}\notag\\
&+&(p^0\rightarrow -p^0),
\label{LWinte}
\end{eqnarray}
where the integration measure ${\rm d}^{D-1}{\bm k}/(2\pi)^{D-1}$ is understood. With $\cD_{\bm k}=\cD_{\bm k}^*=\cD_{\bm k}^\prime=\cD_{\bm k}^{\prime *}=\mathbb{R}^{D-1}$, this is the loop integral as Lee and Wick intended it. Later on, we will deform the integration domains to complex values. The notation used in (\ref{LWinte}) emphasizes that we will do it by keeping the right-hand side real for $p$ real.

We study $\cM(p)$ for real spatial momentum $\bm{p}$ and complex energy $p^0$. If $p^0$ is such that no $\bm{k}$ makes a denominator vanish, the result is a function $\cM(p)$ analytic in $p^0$. If $p^0$ is such that some $\bm{k}$ makes the integrand singular, the function $\cM(p)$ needs not be analytic in $p^0$.

Note that the singularities of the $\bm{k}$ integrand occur where both propagators of ${\cal M}$ are singular, i.e., $(\kl^{2})^{2}+M^{4}=0=[(\kl+\pl)^{2}]^{2}+M^{4}$. (A)\ and (B)\ are the same, upon translating $\bm{k}$ by $-\bm{p}$
and reflecting $\bm{k}$ to $-\bm{k}$. Thus, it is sufficient to study one. Moreover, the singularities of (C)\ and (D)\ have no intersection with the real axis and can be put aside. 

For clarity, it is convenient to study the problem in $D=2$ dimensions, where the computations can be done explicitly \cite{Anselmi:2018kgz}. There, $%
\bm{p}$ and $\bm{k}$ have just one component: $p_{x}$ and $k_{x}$,
respectively. Keeping $p_{x}\in $ $\mathbb{R}$ fixed, the curves drawn by the singularities in the complex $p^{0}$ plane by varying $k_{x}\in $ $\mathbb{R}$ are shown in
figure \ref{fig2}. 
\begin{figure}[t]
	\begin{center}
		\includegraphics[width=14cm]{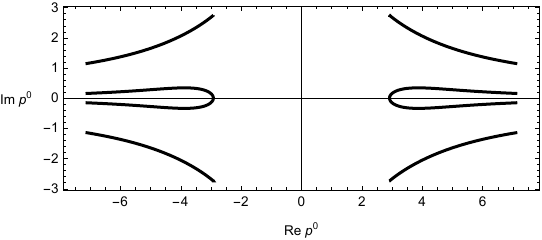}
	\end{center}
	\caption{Non-analyticities of the bubble amplitude ${\cal M}(p)$ in the complex $p^{0}$ plane in $D=2$ dimensions for $k_{x}\in \mathbb{R}$, with $M=2$ and $p_{x}=1$.}
	\label{fig2}
\end{figure}

Explicitly, (A) and (B)\ give the curves%
\begin{eqnarray}
	\gamma &\text{:}&\text{\quad }p^{0}=\sqrt{k_{x}^{2}-\rmi M^{2}}+\sqrt{%
		(k_{x}+p_{x})^{2}+\rmi M^{2}},  \notag \\
	\gamma ^{\prime } &\text{:}&\text{\quad }p^{0}=-\sqrt{k_{x}^{2}-\rmi M^{2}}-%
	\sqrt{(k_{x}+p_{x})^{2}+\rmi M^{2}}. \label{gaga}
\end{eqnarray}%
In figure \ref{fig2}, these are the ones that intersect the real $p^{0}$ axis. Such curves are reminiscent of the branch cuts of an ordinary bubble diagram. The correct cuts, which we identify below by means of a deformation of the integration domain $\cD_{\bm k}$ to complex values, are half lines on the real axis contained in the
regions $\mathcal{A}_{\gamma }$ and $\mathcal{A}_{\gamma ^{\prime }}$
bounded by $\gamma $ and $\gamma ^{\prime }$. The point is that $\mathcal{A}_{\gamma }$ and $\mathcal{A}_{\gamma ^{\prime}}$ are not cuts, but extended regions. Moreover, they are not Lorentz invariant. 

The singularities of (C) and (D)\ give the curves of figure \ref{fig2}
that do not intersect the real axis. Because of this, we do not need to deform the integration domain $\cD_{\bm k}^\prime$ to complex values. 

In $D>2$ dimensions we reach similar conclusions, with the difference that the integrands are singular not just on the curves $\gamma $ and $\gamma ^{\prime }$, but everywhere inside $\mathcal{A}_{\gamma }$ and $\mathcal{A}%
_{\gamma^{\prime}}$. The curves of (C) and (D) also turn into extended regions.

When the external energy $p^0$ is located outside the regions $\mathcal{A}_{\gamma }$ and $\mathcal{A}_{\gamma^{\prime}}$, the function $\cM(p)$ coincides with the one given by the by-the-book prescription. Indeed, when $p^0$ is continued to Euclidean values the LW integration path on $k^0$ is equivalent to the Euclidean loop-energy integration. 

We want to clarify the properties of the function $\cM(p)$ when $p^{0}$ is located inside $\mathcal{A}_{\gamma }$ and $\mathcal{A}%
_{\gamma^{\prime}}$, or at their borders.
In $D>2$ the $\bm{k}$ integral on $\mathbb{R}^{D-1}$ is well defined, but the result
$\cM(p)$ is neither analytic nor Lorentz invariant.
In $D=2$, $\cM(p)$ is both analytic and Lorentz invariant, but  physically correct only for $-\pl^{2}>$ $2M^{2}$. The reason is  that the curve $\gamma$ intersects the real axis in a point that does not have a Lorentz invariant meaning. 

The troubles in both $D>2$ and $D=2$ dimensions are solved by the so-called  \emph{domain deformation}. Instead of integrating $\bm{k}$ on $\mathbb{R}^{D-1}$, we integrate it on a deformed, complex domain $\mathcal{D}_{\bm{k}}$ that shrinks the regions $\mathcal{A}_{\gamma }$ and $\mathcal{A}_{\gamma ^{\prime }}$ into the expected branch cuts $-\pl^2 \geq 2M^2$ on the real axis. Then we find that the result obtained by evaluating ${\cal M}(p)$ with $-\pl^2 \geq 2M^2$ from inside $\mathcal{A}_{\gamma }$ or $\mathcal{A}_{\gamma^{\prime }}$ coincides with the average of the two analytic continuations of the Euclidean $\cM(\pe)$ around the branch point $-\pl^2 = 2M^2$. This operation is what we call \emph{average continuation} of the Euclidean result \cite{Anselmi:2017yux}.

In $D=2$ dimensions (for details in $D>2$, see \cite{Anselmi:2018kgz}), the domain deformation is worked out by inverting eqs.\ (\ref{gaga}) and expressing $k_{x}$ as a function of the external momentum $\pl$. Squaring (to treat both $\gamma $ and $\gamma^{\prime }$ at once) and solving, the solutions for a Lorentzian $\pl^2=-(p^0)^2+p_x^2$ are
\begin{equation}
	k_{x}=-\frac{p_{x}}{2}-\rmi\frac{p_{x}}{\pl^{2}}M^{2}\pm \frac{p^{0}}{2\pl^{2}}\sqrt{%
		(\pl^{2})^{2}-4M^{4}}\,.  \label{defo}
\end{equation}%

When we keep both $p^{0}$ and $p_{x}$ real (because we want to
shrink the whole regions onto the real axis of the complex $p^{0}$ plane), formula
(\ref{defo}) tells us that $k_{x}$ is complex. The two solutions with $\pl$ real
and $-\pl^{2}>2M^{2}$ define the deformed domain $\mathcal{D}_{\bm{k}}$ in
the complex $k_{x}$ plane. It is plotted in figure \ref{fig3}.

The height of the deformation is proportional to $p_{x}$. If we choose a
Lorentz frame with $p_{x}=0$ (which we can do, since $-\pl^{2}>0$ here), the
right-hand side of (\ref{defo}) is real and the regions $\mathcal{A}_{\gamma}$ and $\mathcal{A}_{\gamma ^{\prime }}$ are already shrunk onto branch cuts: no domain deformation is needed. The drawback is that we cannot compute the integral ``from the inside.'' 

The different properties of the computation in different Lorentz frames ($%
p_{x}=0 $, $p_{x}\neq 0$) is a manifestation that Lorentz invariance does
not hold before the deformation. Yet, by construction the AP procedure guarantees that Lorentz invariance is ultimately respected. This means that if our purpose is a mere calculation, rather than studying Lorentz invariance and other general properties, we are allowed to choose a special frame at some intermediate step. In particular, with due caution for integrable singularities, one can fix $p_x=0$ before integrating on the internal spatial momenta. Replacing $p_0^2\to-\pl^2$ in the final expression leads to the same result as in the full procedure \cite{Anselmi:2017yux,Liu:2022gun}. This shortcut to the AP prescription avoids an explicit integration on a complex $\bm{k}$ path.
\begin{figure}[t]
	\begin{center}
	\includegraphics[width=14cm]{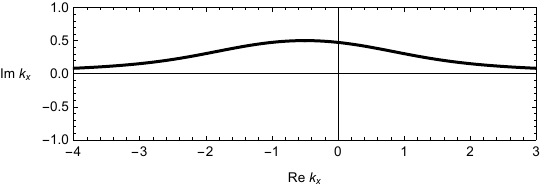}
	\end{center}
	\caption{Deformed integration domain for $k_{x}$ in the $(\Re\,k_x,\Im\, k_x)$ plane, with $M=2$ and $p_{x}=1$. This is the locus of singularities of the integrand of the bubble diagram in the complex $k_x$ plane in $D=2$ dimensions for $p^0\in \mathbb{R}$ and $(p^0)^2>p_x^2+2M^2$, with $M=2$ and $p_{x}=1$.}
	\label{fig3}
\end{figure}

Going back to (\ref{LWinte}) and collecting the pieces of information uncovered
so far, we can interpret the integral $\cM(p)$ in two equivalent
ways, one from ``outside'' and one from
``inside'' the regions $\mathcal{A}_{\gamma
}$ and $\mathcal{A}_{\gamma ^{\prime }}$.
\begin{enumerate}
\item[1.] \textit{Average continuation from the outside}. As we did in the previous section, we first compute (\ref{Mtotapp}), which is the Euclidean version of (\ref{BAP}),
for a Euclidean external momentum $\pe$. Then we
analytically continue the result to the Minkowski domain, with the caveat
that, when we meet the branch cut on the real axis (which occurs for $\pl$ such
that $-\pl^{2}>2M^{2}$), we average the two analytic continuations around the
cut. This way, the regions $\mathcal{A}_{\gamma }$ and $\mathcal{A}_{\gamma ^{\prime }}$
are never met.
\item[2.] \textit{Domain deformation from the inside}. We compute (\ref{LWinte}) by deforming the integration domain in $\bm{k}$, so as to shrink the regions $\mathcal{A}_{\gamma }$ and $\mathcal{A}%
_{\gamma ^{\prime }}$ onto the real axis. We obtain an integral that evaluates precisely to the average continuation from the outside. The proof of this statement is given below. 
\end{enumerate}
The fakeon formulation of purely virtual particles is defined by either procedure. The choice of which one to use practically may depend on the situation at hand.

\begin{theo}
The calculations of $\cM(p)$ by means of 1) the average continuation from the outside and 2) the domain deformation from the inside give the same results.
\end{theo}
\begin{proof}
The statement follows from a symmetry argument, combined with continuity. The first thing to note is that, as long as the regions $\mathcal{A}_{\gamma }$ and $\mathcal{A}%
_{\gamma ^{\prime }}$ have finite extensions, the loop integral evaluates to a continuous function $\cM(p)$,
because, although the integrand is divergent in those regions, the singularity is integrable and one-dimensional.

The second observation is that 
$\cM(p)$
is real for $p$ real, as is clear from
(\ref{LWinte}). Extending $p^0$ to complex values, we have 
\begin{equation}
\cM(p)=\cM^*(p)\qquad \textrm{for every }p.
\label{symmetry}
\end{equation}

Now, take $p^0$ real inside one of the regions $\mathcal{A}_{\gamma}$ and $\mathcal{A}_{\gamma ^{\prime }}$, and calculate the integral $\cM(p)$ there, with $\cD_{\bm{k}}=\cD_{\bm{k}}^\prime=\mathbb{R}^{D-1}$. We know that the result is real.

Shrink each region onto the real axis by deforming the domain $\cD_{\bm{k}}$. Formula (\ref{LWinte}) ensures that (\ref{symmetry}) holds throughout the deformation. Thus, $\cM(p)$ is still real for real $p$ located inside the shrunk regions $\mathcal{A}_{\gamma}$ or $\mathcal{A}_{\gamma ^{\prime }}$. 

Consider the limit of $\cM(p)$ with respect to the shrinking operation, which we denote by ${\rm lim}_{\rm shr}$. What we have just said implies that the imaginary part of the limit ``from the inside'' is zero:
\be
\Im\,({\rm lim}_{\rm shr} \cM)=0\,.
\ee
The real part is discussed below.

We compare this result with what we obtain by  letting $p^0$ tend to the real axis while staying outside the shrinking regions. In this case the result is not affected by the domain deformation, because, by construction, the deformation does not cross any singularity of the integrand. We have two options: reaching the real axis from above and from below. Equation (\ref{symmetry}) ensures that the imaginary part of the limit from above the real axis is equal to minus the imaginary part of the limit from below:
\be
\Im\,({\rm lim}_{\rm shr \downarrow} \cM)=-\Im\,({\rm lim}_{\rm shr \uparrow} \cM)\,.
\ee
Generically speaking, the two limits can be non-zero. Yet, we do not need their values here. We just need to know that they are opposite to each other. Since taking the imaginary part commutes with the limit, we infer that both the imaginary part of the limit of $\cM(p)$ from the inside and the sum of the limits of the imaginary parts from above and from below are zero:
\ben
{\rm lim}_{\rm shr}\,\Im\,\cM=0= {\rm lim}_{\rm shr \downarrow}\Im\,\cM + {\rm lim}_{\rm shr \uparrow} \Im\,\cM\,.
\een
The relationship between the left-most and right-most sides of this expression can only be linear, so we equate them up to an arbitrary coefficient $\a$:
\be\label{cor1}
{\rm lim}_{\rm shr}\,\Im\,\cM=\a\left({\rm lim}_{\rm shr \downarrow}\Im\,\cM + {\rm lim}_{\rm shr \uparrow} \Im\,\cM\right).
\ee

As far as the real parts are concerned, continuity and the symmetry (\ref{symmetry}) ensure that the limit from above coincides with the limit from below, as well as the limit from inside:
\be
{\rm lim}_{\rm shr}\,\Re\,\cM= {\rm lim}_{\rm shr \downarrow}\Re\,\cM={\rm lim}_{\rm shr \uparrow} \Re\,\cM\,.
\ee
Thus, the limit of the real part from the inside coincides with the average of the limits of the real parts from above and from below:
\be
{\rm lim}_{\rm shr}\,\Re\,\cM= \frac12\left({\rm lim}_{\rm shr \downarrow}\Re\,\cM + {\rm lim}_{\rm shr \uparrow} \Re\,\cM\right).
\ee
Comparing this expression with \Eq{cor1}, we can fix $\a=1/2$ and conclude the proof that the limit of $\cM(p)$ from inside one of the regions $\mathcal{A}_{\gamma}$ and $\mathcal{A}_{\gamma ^{\prime }}$ is equal to the average of the two analytic continuations of $\cM(p)$ from above the real axis and from below:
\be
{\rm lim}_{\rm shr}\,\cM= \frac12\left({\rm lim}_{\rm shr \downarrow}\,\cM + {\rm lim}_{\rm shr \uparrow}\,\cM\right),
\label{finall}
\ee
as announced.
\end{proof}

At this point, we know that the result of the domain deformation from the inside, which is the left-hand side of (\ref{finall}), is Lorentz invariant. The line of reasoning is as follows: ($i$) the Euclidean result is Lorentz invariant and coincides with the result of the calculation from outside; ($ii$) the analytic continuations from above and from below the real axis are Lorentz invariant; ($iii$) so is the right-hand side of (\ref{finall}); ($iv$) so is the left-hand side. We stress again that the right-hand side, which is calculated from the outside, is not affected by the domain deformation, which does not cross singularities of the integrand.


\section{Conclusions}\label{sec5}

In this work, we have studied and clarified four different ways to obtain Lorentzian amplitudes in quantum field theories with complex poles: a textbook analytic continuation, the LWN prescription \cite{Lee:1969fy,Lee:1970iw,Nakanishi:1971jj}, the fakeon or AP prescription \cite{Anselmi:2017yux,Anselmi:2017lia,Anselmi:2018bra} and a direct calculation in Lorentzian signature, which fails to yield a viable quantum theory \cite{Aglietti:2016pwz}. Of the above four prescriptions, only AP is physically viable because it respects both the optical theorem and Lorentz invariance.

These four ways need not be the only ones that make sense of field theories with complex poles. We already know one such example, the CLOP procedure \cite{Cutkosky:1969fq}. Both CLOP and the domain deformation advanced here and in \cite{Anselmi:2017yux} deal with the same problem from two different viewpoints. While the former is manifestly Lorentz invariant, the latter recovers such symmetry at a second stage.  Moreover, the former has ambiguities at high orders.
It would be interesting to better clarify the relationship between the two techniques.

The results presented in this paper are not limited to Lee--Wick models and generalize to nonlocal theories which may also have a finite number of poles, such as nonlocal quantum gravity \cite{Modesto:2017sdr,Buoninfante:2022ild,BasiBeneito:2022wux,Koshelev:2023elc} and fractional quantum gravity \cite{BrCa,Calcagni:2021aap,Calcagni:2022shb}. The only difference to take into account is that these theories are typically defined in Euclidean momentum space and then analytically continued \emph{\`a la} Efimov to Lorentzian signature \cite{Efimov:1967dpd,Pius:2016jsl,Briscese:2018oyx,Chin:2018puw,Koshelev:2021orf,Buoninfante:2022krn} (Euclidean external momenta are continued to Lorentzian ones after performing loop integrals). All the above applies with only one modification, namely, that the integration path in the internal energies can now be open (no arc at infinity). The classification of prescriptions for scattering amplitudes presented here holds in these perturbative approaches to quantum gravity and, in particular, puts unitarity results in fractional gravity on a more solid ground. 


\section*{Acknowledgments}

F.B., G.C.\ and L.M.\ are supported by grant PID2023-149018NB-C41 (G.C.\ as PI) funded by the Spanish Ministry of Science, Innovation and Universities MCIN/AEI/10.13039/ 501100011033. F.B. is grateful to P.M. Santini for useful discussions on Appendix C.


\appendix


\section{Troubles with the direct Minkowski prescription}\label{appA}

In this appendix, we recall a result of \cite{Aglietti:2016pwz}, stating that loop integrals of higher-derivative theories with complex poles have nonlocal divergent parts if they are defined directly in Minkowski spacetime. We study once again the bubble diagram of the model (\ref{LWmodel}) but take both internal and external momenta $\kl$ and $\pl$ to be Lorentzian from the start. 

The amplitude reads
\be
\cM_{\rm dM}(p^2) = \frac{\rmi \la^2}{2}\int_{\mathbb{R}^D} \frac{\rmd^D k}{(2\pi)^D}\,\tilde \cG(-k^2)\,\tilde \cG[-(k+p)^2]\,,\label{MtotappA}
\ee
where ``dM'' stands for direct Minkowski. Clearly, $\cM_{\rm dM}(p^2)$ is Lorentz invariant. Moreover, it is real for real $p$, which is good for the optical theorem. The problem is that it does not satisfy the locality of counterterms.

The common rules of power counting do not hold in the direct Minkowski approach because the propagator does not fall off rapidly enough along light cones.  Moreover, since the propagators have poles located in the first and third quadrants of the complex loop energy plane, the result is not straightforwardly related to the Euclidean one. Hence, it does not inherit the usual power-counting behaviour of the Euclidean integral. 

The consequence is that the direct Minkowski approach violates the locality of counterterms and the theory is not BPHZ renormalizable. For example, the integral (\ref{MtotappA}) has a nonlocal divergent part in $D=6$ dimensions. Nonlocal divergent parts also appear in $D=4$ if the integrand is multiplied by polynomials of $k^\mu$, brought in by non-trivial vertices, as in gravity theories.

We briefly describe what happens in $D=6$ dimensions, directing the reader to \cite{Aglietti:2016pwz} for the case $D=4$. We first integrate on the loop energy $k^0$ by means of the residue theorem. Then we expand the integrand for large $|{\bm k}|$. Finally, we integrate term by term in ${\bm k}$, using dimensional regularization. At the end, we obtain the nonlocal divergent part 
\be
\cM_{\rm dM}^\textrm{div}(p^2)=\frac{2\rmi M^2\lambda^2}{3\,\ve (4\pi)^3}\left[\frac{M^4}{(p^2)^2}-\frac{3}{4}\right],
\ee
where $D=6-\ve$.


\section{Optical theorem with complex poles}\label{appB}

In this appendix, we elaborate on the implications of the optical theorem for theories with complex poles. Decomposing the $S$ matrix as usual, $S=1+\rmi T$, the optical theorem $-\rmi T+\rmi T^\dagger=TT^\dagger$ is another way to write the unitarity equation $SS^\dagger=1$. 

The optical theorem can be expressed diagrammatically by means Cutkosky-Veltman diagrams, also knows as ``cut diagrams'', which are made of two parts, separated by a cut. One part corresponds to $iT$ and the other one (here denoted with an asterisk) corresponds to $-iT^\dagger$.

In the instances of interest to us, which are the propagator and the bubble diagram, the diagrammatic identities read
\ba
2\hspace{0.01in}\text{Im}\left[ (-\rmi)%
\raisebox{-1mm}{\scalebox{2}{$\rangle
\hspace{-0.07in}-\hspace{-0.07in}\langle$}}\,\right] =%
\raisebox{-1mm}{\scalebox{2}{$\rangle
\hspace{-0.07in}-\hspace{-0.13in}\slash\hspace{-0.033in}_{\raisebox{-2pt}{\scalebox{0.5}*}}\hspace{-0.025in}\langle$}}+\raisebox{-1mm}{\scalebox{2}{$\rangle
\hspace{-0.07in}-\hspace{-0.13in}^{\scalebox{0.5}*}\hspace{-0.043in}\slash\hspace{-0.015in}\langle$}} &=&\int 
\mathrm{d}\Pi _{\rm f}\hspace{0.01in}\left\vert \raisebox{-1mm}{\scalebox{2}{$%
\rangle\hspace{-0.035in}-$}}\right\vert ^{2},  \label{cutd} \\
2\hspace{0.01in}\text{Im}\left[ (-\rmi)\raisebox{-1mm}{\scalebox{2}{$-%
\hspace{-0.055in}\bigcirc\hspace{-0.055in}-$}}\right] =\hspace{0.01in}%
\raisebox{-1mm}{\scalebox{2}{$-\hspace{-0.055in}\bigcirc\hspace{-0.15in}%
\slash\hspace{0.015in}-$}}\hspace{-0.2in}\raisebox{-11.6pt}{*}\hspace{0.11in}+\hspace{0.01in}%
\hspace{0.11in}\raisebox{6pt}{*}\hspace{-0.2in}\raisebox{-1mm}{\scalebox{2}{$-\hspace{-0.055in}\bigcirc\hspace{-0.15in}%
\slash\hspace{0.015in}-$}} &=&\int \mathrm{d}\Pi _{\rm f}\hspace{0.01in}%
\left\vert \raisebox{-1mm}{\scalebox{2}{$-\hspace{-0.035in}\langle$}}%
\right\vert ^{2}.  \label{optapp}
\ea
The vertical slash crossing propagators is the cut that divides the diagram in the two portions just mentioned. The portion with an asterisk is built with complex conjugate vertices and propagators. It is understood that positive energies flow from the portion without asterisks into the portion with the asterisk. 

The right-hand sides of the two equations represent the phase-space integrals involved in the cross sections for the production of the particles associated with the cut propagators. 
Such particles, and their cut propagators, can be derived from (\ref{cutd}). Inserting the cut propagators into (\ref{optapp}), it is then possible to check the optical theorem.

For example, in the case of the Feynman prescription, we have
\be
2\,\hspace{0.01in}\text{Im}\left[(-\rmi)(-1)\frac{-\rmi}{k^2+m^2-\rmi\e}\right]=\theta(k^0)(2\pi)\delta(k^2+m^2)+\theta(-k^0)(2\pi)\delta(k^2+m^2)\,,
\ee
where the factor $(-1)$ is due to the vertices (setting the coupling to one). On the right-hand side, the product between a vertex and its complex conjugate gives $+1$. 

We infer that the cut propagator of a physical particle is
\be
\raisebox{-1mm}{\scalebox{2}{$
-\hspace{-0.06in}-\hspace{-0.18in}\slash\hspace{-0.033in}_{\raisebox{-2pt}{\scalebox{0.5}*}}$}}\hspace{.12in}=\theta(k^0)(2\pi)\delta(k^2+m^2)\,.
\label{cutproph}
\ee

Now we check (\ref{optapp}) in the case of the bubble diagram with circulating physical particles in the massless limit. Formula (\ref{Mph}) gives
\be
2\,\textrm{Im}\cM_{\textrm{bb}}(p^2)= \frac{\theta(-p^2)}{16\pi}\,.
\ee
The cut diagrams of \Eq{optapp} are bubbles calculated with the cut propagator (\ref{cutproph}). For example, the first one gives
\ba
\raisebox{-1mm}{\scalebox{2}{$-\hspace{-0.055in}\bigcirc\hspace{-0.15in}%
\slash\hspace{0.015in}-$}}\hspace{-0.2in}\raisebox{-11.6pt}{*}\hspace{0.11in}&=&\frac{1}{2}\int \frac{\rmd^Dk}{(2\pi)^D}\theta(k^0)(2\pi)\delta(k^2)\theta(p^0-k^0)(2\pi)\delta[(p-k)^2]\nn
&=&\frac{\theta(p^0)\,\theta(-p^2)}{16\pi}+O(\ve)\,.
\ea
The second one gives the same with $\theta(p^0)\rightarrow\theta(-p^0)$.
We see that, in total, (\ref{optapp}) holds in $D=4$ dimensions for standard particles.

In the case of a propagator like (\ref{HDprop}), the left-hand side of equation (\ref{cutd}) vanishes. This implies that the cut propagators vanish. Hence, the left-hand side of (\ref{optapp}) also vanishes and (\ref{HDprop}) must be treated as fake particles. To make the right-hand sides of (\ref{cutd}) and (\ref{optapp}) disappear as well, the modes associated with the complex conjugate poles must be removed from the physical spectrum of asymptotic states. This ``projection," combined with the average continuation (or the LW integration path for loop energies plus the domain deformation for loop spatial momenta), completes the AP prescription for fake particles and ensures consistency with the optical theorem.


\section{Massive \texorpdfstring{$\phi^3$}{phi3} bubble amplitude}\label{appC}

In this appendix, we calculate the Euclidean bubble amplitude for an interacting cubic scalar field theory and study ways to extend the result to Lorentzian signature.


\subsection{Calculation of \Eq{final}}\label{appC1}

The one-loop  bubble amplitude for the model \Eq{LWmodel} is
\ba
\cM(\pe^2) &=& \frac{\rmi \la^2}{2}\int_{\mathbb{R}^D} \frac{\rmi\,\rmd^D \ke}{(2\pi)^D}\,\tilde \cG(-\ke^2)\,\tilde \cG[-(\ke+\pe)^2]\,,\label{Mtotapp}
\ea
where both $\ke$ and $\pe$ are Euclidean and the second $\rmi$ is due to the analytic continuation from Lorentzian to Euclidean internal momentum. Using Feynman parameters (see, e.g., \cite{PeSc}), we find
\be
\cM(\pe^2)\eqqcolon \cM_{++}(\pe^2) + \cM_{+-}(\pe^2) + \cM_{+-}^*(\pe^2)+ \cM_{++}^*(\pe^2)\,,
\ee
where
\ba
\cM_{++}(\pe^2)&=&\frac{\la^2}{8}\int_{\mathbb{R}^D} \frac{\rmd^D \ke}{(2\pi)^D}\,\frac{1}{\ke^2+m^2+\rmi M^2}\frac{1}{(\ke+\pe)^2+m^2+\rmi M^2}\nn
&=&\frac{\la^2}{8}\frac{\Gamma\left(2-\frac{D}{2}\right)}{(4\pi)^{D/2}}\int_0^1 \rmd x\,\left[\pe^2 x(1-x)+m^2+iM^2\right]^{\frac{D}{2}-2}\nn
&=&\frac{\la^2}{128\pi^2}\left\{
\frac{2}{\varepsilon}+\ln(4\pi)-\g_{\textsc{em}}-\int_0^1 \rmd x\,\ln\left[\pe^2 x(1-x)+m^2+\rmi M^2\right]+\cO (\varepsilon)
\right\},\nn\\
\cM_{+-}(\pe^2)&=&-\frac{\la^2}{8}\int_{\mathbb{R}^D} \frac{\rmd^D \ke}{(2\pi)^D}\,\frac{1}{\ke^2+m^2+\rmi M^2}\frac{1}{(\ke+\pe)^2+m^2-\rmi M^2}\nn
&=&-\frac{\la^2}{8}\frac{\Gamma\left(2-\frac{D}{2}\right)}{(4\pi)^{D/2}}\int_0^1 \rmd x\,\left[\pe^2 x(1-x)+m^2-\rmi M^2(1-2x)\right]^{\frac{D}{2}-2}\nn
&=&-\frac{\la^2}{128\pi^2}\left\{
\frac{2}{\varepsilon}+\ln(4\pi)-\g_{\textsc{em}}\right.\nn
&&\left.\qquad\qquad
-\int_0^1 \rmd x\,\ln\left[\pe^2 x(1-x)+m^2-\rmi M^2(1-2x)\right]+\cO (\varepsilon)
\right\}.
\ea
Here we have expanded $D=4-\ve$ for small $\ve$.

Summing the various contributions and sending $\ve\to 0$, we find
\be\label{finale}
\cM(\pe^2)=\frac{\la^2}{128\pi^2}\int_0^1\rmd x\,\ln\frac{[\pe^2x(1-x)+m^2]^2+(1-2x)^2M^4}{[\pe^2 x(1-x)+m^2]^2+M^4}\,,
\ee
in $D=4$ dimensions. We emphasize that the complex logarithm in the integral is meant as its principal branch, since (\ref{Mtotapp}) is real.

\subsection{Analytic and average continuations of \Eq{final}}\label{appC2}

In this section, we study the analytic continuation of (\ref{final}), then perform the average continuation and finally show the differences with respect to (\ref{finalL}).
	First, we define the following complex function
\be\label{final complex}
\cM(z)=\frac{\la^2}{128\pi^2}\int_0^1\rmd x\,\ln\frac{[x(1-x) z+\left(m/M\right)^2]^2+(1-2x)^2}{[x(1-x) z +\left(m/M\right)^2]^2+1}\, , \qquad \Re \,z > 0\, ,
\ee
where $z$ is a complex variable, so that (\ref{final}) corresponds to a real positive $z= \pe^2/M^2$.	Note the different mass scaling of $z$ with respect to the convention used in the main body of the paper.

To simplify the task without losing key properties, we choose $m=0$. The first contribution to (\ref{final complex}) is given by the function
	\ba
	f_+(z)&=&\int_{0}^{1}\rmd x\,\ln \left[x^{2}(1-x)^{2}z^{2}+(1-2x)^{2}\right]\label{inte}\\
	&=&  \int_{0}^{1}\rmd x\left\{\ln \left[x(1-x)z+\rmi(1-2x)\right]+\ln\left[x(1-x)z-\rmi(1-2x)\right]\right\}\label{inte0},  \qquad  \Re \,z > 0\,. \nonumber 
	\ea
An explicit evaluation  gives 
	\be
	f_{+}(z)= \frac{2\pi }{z}-4+\rmi\frac{\sqrt{4-z^{2}}}{z}\, \textrm{Ln} \left[\frac{z+\rmi\sqrt{			4-z^{2}}}{z-\rmi\sqrt{4-z^{2}}}\right]\, ,  \qquad  \Re \,z > 0\,,   \label{fz}
	\ee
where we have used the notation $\textrm{Ln}(z)$ to emphasize that the logarithm is intended in its principal branch. We also define the square roots as their principal branches, even if choosing the other branches would not change (\ref{fz}).
	
Let us study the analytic continuation of (\ref{fz}) to the whole complex $z$ plane. We first show that (\ref{fz}) has a simple zero (not a branch point) at $z=2$. 	We use the following expansion
	\begin{equation}\label{c9}
\textrm{Ln}(z+\mu) = \textrm{Ln}(z) + \sum_{j=1}^{\infty}\frac{\left(-1\right)^{j-1}}{j}\left(\frac{\mu}{z}\right)^j \, \quad \text{for} \quad \left\vert\frac{\mu}{z} \right\vert < 1, 
	\end{equation}
where $\mu = \sqrt{4-z^{2}}$.	The condition $\left\vert\mu^2/z^2 \right\vert=\left\vert\left( 4-z^{2}\right)/z^2 \right\vert < 1$ is verified in two disjoint regions of the complex $z$ plane, the	region $\Re\, z > \sqrt{2+ (\Im\, z)^2 }$ depicted in blue in figure~\ref{fig4} and the region $\Re\, z <- \sqrt{2+ (\Im\, z)^2 }$ depicted in red. 

Indeed, in the neighborhood of $z=2$ one finds the expansion
\be\label{series1}
i\,\sqrt{4-z^{2}}\, \textrm{Ln} \left[\frac{z+\rmi\sqrt{	4-z^{2}}}{z-\rmi\sqrt{4-z^{2}}}\right] = \sum_{j=0}^{\infty}\frac{2 \, z}{2j+1}\left(\frac{z^2-4}{z^2}\right)^{j+1}  ,\quad \Re\, z > \sqrt{2+ (\Im\, z)^2 }\,,
\ee
which shows that $z=2$ is regular and it is a simple zero of the function $f_+(z)$ defined in (\ref{fz}).

Now, we want to study the analytic continuation of $f_+(z)$ to the region $\Re \, z <0$ of the complex $z$ plane. It is obtained by the analytic continuation of the $\textrm{Ln}$ in \Eqq{fz}, that is,
\be\label{Ln 1}
\ln \left[\frac{z+\rmi\sqrt{	4-z^{2}}}{z-\rmi\sqrt{4-z^{2}}}\right] = \left\{\begin{array}{ll}
	\textrm{Ln} \left[\dfrac{z+\rmi\sqrt{	4-z^{2}}}{z-\rmi\sqrt{4-z^{2}}}\right], \qquad\qquad \Re\, z >0\\
	\\
	\textrm{Ln} \left[\dfrac{z+\rmi\sqrt{	4-z^{2}}}{z-\rmi\sqrt{4-z^{2}}}\right]+ 2\pi\rmi\,, \qquad \Re\, z \leq 0
\end{array}\right.\,.
\ee
In fact, the argument of the logarithm is real negative if, and only if, $z$ is purely imaginary and the two expressions on the right-hand side of (\ref{Ln 1}) are the analytic continuation of each other, as they coincide on the line $\Re \, z=0$.

The analytic continuation $f_\textrm{an}(z)$ of $f_{+}(z)$ to the region $\Re\,z<0$ is
\ba
\hspace{-1cm}f_{\textrm{an}}(z)&\coloneqq& \frac{2\pi }{z}-4+\rmi\frac{\sqrt{4-z^{2}}}{z}\, \ln \left[\frac{z+\rmi\sqrt{			4-z^{2}}}{z-\rmi\sqrt{4-z^{2}}}\right]\nn
&&\vphantom{1}\nn
&=&\left\{ 
\begin{array}{ll}
\dfrac{2\pi }{z}-4+\rmi\dfrac{\sqrt{4-z^{2}}}{z}\, \textrm{Ln}	  \left[\dfrac{z+\rmi\sqrt{			4-z^{2}}}{z-\rmi\sqrt{4-z^{2}}}\right] , \qquad\qquad\qquad\,\,\,\, \Re\, z > 0\\ 
	\\
	\dfrac{2\pi }{z}-4+\rmi\dfrac{\sqrt{4-z^{2}}}{z}\, \left\{\textrm{Ln}	  \left[\dfrac{z+\rmi\sqrt{			4-z^{2}}}{z-\rmi\sqrt{4-z^{2}}}\right]+2\pi\rmi \right\}, \qquad \Re\, z \leq 0
\end{array}
\right. .
\label{fan}
\ea
Expanding (\ref{Ln 1}) in $z=0$, one has
\be\label{series2}
\ln \left[\frac{z+\rmi\sqrt{	4-z^{2}}}{z-\rmi\sqrt{4-z^{2}}}\right] =\ln \left[-1+\frac{z^2}{2}+\frac{\rmi z \sqrt{4-z^{2}}}{2}\right] = \pi\rmi + O(z^2) \, ,\qquad \vert z \vert \ll 1 ,
\ee
which implies that $z=0$ is a regular point of $f_ {\textrm{an}}(z)$. Moreover, from (\ref{fan}) it is also evident that $z=-2$ is a square-root branch point of $f_ {\textrm{an}}(z)$. Indeed, $f_ {\textrm{an}}(z)$ is analytic on the complex $z$ plane cut along the semi-line $(-\infty,-2]$.
	\begin{figure}[t]
		\begin{center}
		\includegraphics[width=10truecm]{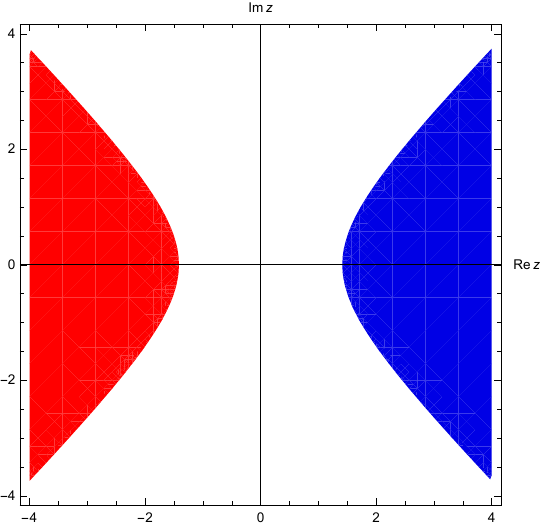}
		\end{center}
		\caption{The regions $\Re\, z > \sqrt{2+ (\Im\, z)^2 }$ (blue) and $\Re\, z <- \sqrt{2+ (\Im\, z)^2 }$ (red) of the complex $z$ plane which solve the condition $\left\vert\left( 4-z^{2}\right)/z^2 \right\vert < 1$.}
		\label{fig4}
	\end{figure}
	
For completeness, we give the explicit expression of $f_{\textrm{an}}(z)$ for real values of its argument, i.e., $z=x \in \mathbb{R}$:
\be\begin{array}{ll}
	f_{\textrm{an}}(x)
	=\left\{ 
	\begin{array}{ll}
		\dfrac{2\pi}{x}-4+ \dfrac{\sqrt{x^{2}-4}}{x}\, \textrm{Ln}	  \left[\dfrac{x+\sqrt{			x^{2}-4}}{x-\sqrt{x^{2}-4}}\right] , \qquad\qquad\qquad\qquad x\geq  2\\ 
		\\
		\dfrac{2\pi }{x}-4-2\,\dfrac{\sqrt{4-x^{2}}}{x}\, \arctan\left[\dfrac{\sqrt{4-x^{2}}}{x}\right] , \qquad\qquad\qquad 0 < x \leq 2
\\
\\
		\dfrac{2\pi }{x}-4-2\,\dfrac{\sqrt{4-x^{2}}}{x}\left\{\pi + \arctan\left[\dfrac{\sqrt{4-x^{2}}}{x}\right] \right\}, \qquad  -2 < x <0
	\end{array}
	\right. .
\end{array}
\label{fan 2}
\ee
Moreover, 
\be
	f_{\textrm{an}}(x \pm i \epsilon)
	=		\frac{2\pi }{x}-4+ \frac{\sqrt{x^{2}-4}}{x}\, 	 \left\{ \textrm{Ln}\left[\frac{x+\sqrt{			x^{2}-4}}{x-\sqrt{x^{2}-4}}\right] \mp 2 \pi i \right\} , \qquad x < - 2 .
\ee

The second contribution to (\ref{final complex}) is encoded in the function
\ba
g(z)&=&\int_{0}^{1}\rmd x\,\ln \left[x^{2}(1-x)^{2}z^{2}+1\right]\label{fan 3}\\
&=&  \int_{0}^{1}\rmd x\left\{\ln \left[x(1-x)z+\rmi \right]+\ln\left[x(1-x)z-\rmi \right]\right\}\label{inte00},  \qquad  \Re \,z > 0 \nonumber 	
\ea
which can be evaluated to 
\be
\begin{array}{ll}
	g(z)= -4 + g_+(z)+ g_-(z),\quad \text{with}\\
	\\
	g_\pm(z) = \sqrt{\dfrac{z\pm 4 \rmi}{z}} \,\,\textrm{Ln} \left[\dfrac{1+\sqrt{z/\left(z\pm 4 i\right)}}{1-\sqrt{z/\left(z\pm 4 i\right)}}\right] \, , \qquad  \Re \,z > 0.
\end{array}
\ee
Using the same expansion as in \Eq{c9},
 one has
\be
	g_\pm(z)= 2 \sum_{j=0}^{\infty} \frac{1}{1+2j} \left(\frac{z}{z\pm 4 i}\right)^{ j}, \qquad \vert \Im\, z \vert <2 \,.
\ee
Therefore, the analytic continuation $g_{\textrm{an}}(z)$ of (\ref{fan 3}) is regular at $z=0$, and has two branch points at $z= \pm 4 i$. Moreover, the representation 
\be
g_{\textrm{an}}(z)= -4+ 2 \sum_{j=0}^{\infty} \frac{1}{1+2j} \left[\left(\frac{z}{z+ 4 i}\right)^{ j}+\left(\frac{z}{z- 4 i}\right)^{ j}\right]\,, \qquad \vert \Im\, z \vert <2 \, ,
\ee 
is valid on a strip of the complex $z$ plane that includes the real axis. Furthermore, one has $g_{\textrm{an}}(z)= g_{\textrm{an}}(-z)$, and $g^*_{\textrm{an}}(z)=g_{\textrm{an}}(z^*)$, so $g_{\textrm{an}}(z)$ is symmetric under reflection with respect the imaginary axis of the complex $z$ plane, and it is real for real $z$.

In conclusion, the analytic continuation of (\ref{final complex}), that is to say,
\be
\cM_{\textrm{an}}(z)=\frac{\la^2}{128\pi^2} \left[f_{\textrm{an}}(z)-g_{\textrm{an}}(z)\right],
\ee
has two branch points at $z=\pm 4 i$ and a branch point at $z=-2$. Moreover, it is real for real $z>-2$, while it has an imaginary part for real $z<-2$, given by
\be
\Im \,\cM_{\textrm{an}}(x \pm \rmi \epsilon)
=\mp   \frac{\la^2}{64\pi} \frac{\sqrt{x^{2}-4}}{x}\,, \qquad x < - 2\, .
\ee
The branch point at $z=-2$ corresponds to a threshold of production of intermediate real states.

As an alternative to the analytic continuation of (\ref{final complex}), one can define the average continuation of (\ref{final complex}) for real $z<-2$ as in \Eqq{APmain}, that is,
\be
\cM_{\textrm{AP}}(x)=\frac{\la^2}{256\pi^2}\left[f_{\textrm{an}}(x-\rmi\e)+f_{\textrm{an}}(x+\rmi\e)-2g_{\textrm{an}}(x)\right], \qquad x\in\mathbb{R}\,,
\label{MAP}
\ee
which cancels the imaginary part of $f_{\textrm{an}}(x)$ for $x<-2$. This gives a complex scattering amplitude that is real for any real $x$,  
which agrees with the optical theorem for purely virtual particles. 

Expression \Eq{finalL}, instead, consists of the formal replacement $\pe^2 \rightarrow p^2$ inside the integrand of (\ref{final}).
In the case at hand ($m=0$) we obtain
\be\label{amplitude fM}
\cM_{\textrm{fM}}(z)= \left\{ \begin{array}{ll}\cM_{\textrm{an}}(z)\,, \,\,\,\,\qquad \Re\, z >0\\
	\\
\cM_{\textrm{an}}(-z)\,, \qquad \Re\, z < 0
\end{array}\right..
\ee
This function satisfies $\cM_{\textrm{fM}}(z)=\cM_{\textrm{fM}}(-z)$ but is not analytic across the imaginary axis of the complex $z$ plane. Restricting to real $z=x \in \mathbb{R}$, (\ref{amplitude fM}) is real and has a corner in $x=0$, i.e., it is continuous but not $C^1$ (discontinuous derivatives) therein. The physical meaning of this behaviour is obscure at the moment.


\section{Lorentz violation with real masses: mixed momenta prescription}\label{appD}

In this appendix, we present a simple example that illustrates how easy it is to violate Lorentz symmetry. We consider the ordinary bubble diagram with propagator
\be\label{propc2}
\tilde G(-\kl^2,m^2) = -\frac{\rmi}{\kl^2+m^2}\,,\qquad m\in\mathbb{R}\,.
\ee
However, instead of evaluating the integral in the Euclidean domain and then analytically continuing the result to Minkowski spacetime, we choose hybrid momenta: a Euclidean internal momentum in combination with a Lorentzian external momentum.

Specifically, we study 
\begin{equation}\label{i}
\mathcal{M}(\pl)=\la^2\int \frac{\rmd^{D}\ke}{(2\pi)^{D}}\frac{1}{\ke^2+m^2}\frac{1}{%
(\ke+\pL)^2+m^2}\,,  
\end{equation}%
where the integrated momentum $\ke=(k_{D},\bm{k})$ is Euclidean, the external momentum $\pl=(p^{0},\bm{p})$ is Lorentzian and $\pL\coloneqq(-\rmi p^{0},\bm{p})$ denotes the Wick rotation of the Euclidean external momentum $\pe=(p_D,\bm{p})$. We want to show that $\mathcal{M}$ is not Lorentz invariant, which is expected because \Eq{i} has the same structure as \Eq{eq12}, where the imaginary contribution comes from the $2k_D\pL^0=-2\rmi k_D p^{0}$ term  in $2\,\ke\cdot\pL$. For definiteness, we assume $p^{0}>0$, $\pl^2<0$.

Let us write $\pL=Wp$, where $W=\textrm{diag}(-\rmi,1,1,1)$. We consider a Lorentz transformation $p\rightarrow\Lambda p$ and combine it with a change of variables given by a rotation $R$. The only ingredient of (\ref{i}) that is not manifestly invariant is the square $(\ke+\pL)^2$. It is easy to check that $(\ke+\pL)^2$ is invariant for arbitrary $\ke$ and $\pL$ if, and only if,
\be
R=W\Lambda W^*.
\ee
In particular, the rotation $R$ must be complex, which is not a legitimate change of variables in the integral. This signals that the result of (\ref{i}) likely violates Lorentz invariance.

To show that indeed it does, we compute it explicitly and compare the result with the one of the usual bubble diagram, which is obviously Lorentz invariant. In Euclidean spacetime, the bubble integral is
\begin{equation}
\mathcal{M}(\pe)=\la^2\int \frac{\rmd^{D}\ke}{(2\pi)^{D}}\frac{1}{\ke^2+m^2}\frac{1}{%
(\ke+\pe)^2+m^2},  \label{ip}
\end{equation}%
where $\pe=(p_{D},\bm{p})$ is also Euclidean. Then we have to perform the Wick
rotation. We denote the result of this operation by $\mathcal{M}_{\rm an}(\pl)$.

While $\mathcal{M}(\pl)$
does not need any prescription because its integrand is never singular, $\mathcal{M}_{\rm an}(\pl)$ needs one for the Wick rotation. We take
this into account by assuming that $p^{0}$ has a small positive imaginary part $\epsilon$. 

The functions $\mathcal{M}(\pl)$ and $\mathcal{M}_{
\rm an}(\pl)$ do not coincide. To see this, we integrate on the loop energies $k_{D}$ by means of the residue theorem. 

In both cases, we have poles at $k_D=\bar k_{D,1}^{\pm}\equiv\pm \rmi\omega_{\bm{k}}$, where $\omega_{\bm{q}}=\sqrt{{\bm q}^2+m^2}$. Then we
have poles at $k_D=\bar k_{D,2}^{\pm}\equiv\rmi p^{0}\pm \rmi\omega_{\bm{k+p}}$ for $\mathcal{M}(\pl)$ and poles at $k_D=\tilde k_{D,2}^{\pm}\equiv -p_{D}\pm \rmi\omega_{\bm{k+p}}$ for $%
\mathcal{M}(\pe)$. Closing the $k_{D}$ integral by means of an arc that crosses the positive real $k^{0}$ axis, where $k^{0}=\rmi k_{D}$, both $%
\mathcal{M}(\pl)$ and $\mathcal{M}(\pe)$ pick the residue at the pole 
$k_D=\bar k_{D,1}^-=-\rmi\omega_{\bm{k}}$. In addition, $\mathcal{M}(\pe)$ picks the residue at the pole $k_{D}=\tilde k_{D,2}^-=-p_{D}-\rmi\omega_{\bm{k+p}}$. Instead, $\mathcal{M}(\pl)$ picks the one at $k_{D}=\bar k_{D,2}^-=\rmi p^{0}-\rmi\omega_{\bm{k+p}}$, but only if $%
\omega_{\bm{k+p}}>p^{0}$. After the calculation of ${\cal M}(\pe)$, the analytic continuation to ${\cal M}_{\rm an}(\pl)$ is done by replacing $p_D$ with $-\rmi(p^0+\rmi\epsilon)$.

We focus on the combination%
\begin{equation}
\mathcal{U}(\pl)\coloneqq \mathcal{M}(\pl)-\frac{1}{2}%
\mathcal{M}_{\rm an}(\pl)+\frac{1}{2}\overline{\mathcal{M}_{\rm an}(\pl)}.
\end{equation}%
Since $\mathcal{M}_{\rm an}(\pl)$ is Lorentz invariant, it is sufficient to prove that $\mathcal{U}(\pl)$ is not.

After some work, we obtain
\begin{equation*}
\mathcal{U}(\pl)=\la^2\,\mathrm{PV}\int \frac{\rmd^{D-1}\bm{k}}{%
(2\pi)^{D-1}}\frac{\theta (p^{0}-\omega _{\bm{k+p}})}{4\omega _{\bm{k}}\omega _{\bm{k+p}}}\left(\frac{1}{p^{0}-\omega _{\bm{k+p}%
}-\omega _{\bm{k}}}-\frac{1}{p^{0}-\omega_{\bm{k+p}}+\omega _{\bm{k}}}\right),
\end{equation*}%
where $\mathrm{PV}$ denotes the Cauchy principal value. 

The integral is convergent in $D=4$. After a translation $\bm{k}%
\rightarrow \bm{k-p}$, we can evaluate it straightforwardly at $m=0$.
The result is
\begin{equation}
\mathcal{U}(\pl)=\frac{\la^2}{8\pi ^{2}}\mathrm{PV}\int_{-1}^{1}\rmd u\int_{0}^{p^{0}}\frac{\rmd k\,k}{(p^{0}-k)^{2}-k^{2}-\bm{p}^{2}+2k|\bm{p}|u}=-\frac{\la^2}{8\pi^2}\frac{p^{0}}{|\bm{p}|}\,\tanh^{-1}\frac{|\bm{p}|}{p^{0}}\,,
\end{equation}%
where $k=|\bm{k}|$. Writing $p^{0}=\sqrt{s+\bm{p}^{2}}$, we find%
\begin{equation}
\mathcal{U}(\pl)=-\frac{\la^2}{24\pi^2}\left[3+\frac{\bm{p}^{2}}{s}+%
\mathcal{O}(|\bm{p}|^{4})\right] .
\end{equation}%
We see that $\mathcal{U}(\pl)$ is not Lorentz invariant, because it does not depend on the invariant $s$ only, but also on the square of the spatial external momentum ${\bm p}$.


\end{document}